\documentclass{llncs}

\usepackage{amsmath}
\usepackage{amssymb}
\usepackage{bm}
\usepackage{lmodern,microtype}
\usepackage[utf8]{inputenc}
\usepackage[T1]{fontenc}
\usepackage{braket}       % Dirac notation: \ket, \bra, \braket
\usepackage{mathrsfs}     % \mathscr
\usepackage{algorithm}
\usepackage{algpseudocode}
\algrenewcommand\algorithmicrequire{\textbf{Input:}}
\algrenewcommand\algorithmicensure{\textbf{Output:}}

\usepackage{thmtools} 
\usepackage{thm-restate}

\usepackage[hidelinks]{hyperref}
\usepackage{cleveref}
\usepackage{tikz}
\usetikzlibrary{positioning}
\usepackage{todonotes}
\newcommand{\ZZ}{\mathbb{Z}}
\newcommand{\CC}{\mathbb{C}}
\newcommand{\RR}{\mathbb{R}}
\newcommand{\QQ}{\mathbb{Q}}
\newcommand{\FF}{\mathbb{F}}
\newcommand{\calU}{\mathcal{U}}
\newcommand{\calL}{\mathcal{L}}

\newcommand{\zetap}{\zeta_p}

\newcommand{\inner}[2]{\langle #1,\, #2 \rangle}
\newcommand{\cyc}{R}
\newcommand{\cycq}{R_q}
\newcommand{\Tr}{\operatorname{Tr}}
\newcommand{\norm}[1]{\lVert #1 \rVert }
\newcommand{\ring}[1]{R_{#1}}
\newcommand{\CCP}{\operatorname{CCP}}
\newcommand{\DCP}{\operatorname{DCP}}
\newcommand{\LWE}{\operatorname{LWE}}
\newcommand{\CyLWE}{\operatorname{CyLWE}}
\newcommand{\PSP}{\operatorname{PSP}}
\newcommand{\poly}{\operatorname{poly}}
\newcommand{\x}{\mathbf{x}}
\newcommand{\y}{\mathbf{y}}
\newcommand{\s}{\mathbf{s}}

\renewcommand{\a}{\mathbf{a}}

\newcommand{\yixin}[1]{{
  \color{purple}Yixin: #1
  }}

\newcommand{\TODO}[1]{{
  \color{olive}TODO: #1
  }}

\begin{document}

\title{Cyclotomic Cosets: Hidden Subgroup and Quantum Sieving Algorithm for Prime-Power Moduli}
\author{Mathias Boucher \and Pierre-Alain Fouque \and Yixin Shen}
\institute{Univ Rennes, Inria, CNRS, IRISA, Rennes, France}
\maketitle
\pretolerance=100\tolerance=2000\emergencystretch=3em

\begin{abstract}
The Learning With Errors (LWE) problem is a fundamental assumption in post-quantum cryptography. Regev established a quantum reduction from LWE to the Dihedral Coset Problem (DCP). Later, Brakerski et al. introduced the Extrapolated Dihedral Coset Problem (EDCP), proving its equivalence to LWE. However, unlike DCP, EDCP no longer admits a coset structure. This limits the direct application of \mbox{techniques for hidden subgroup problems.}

In this work, we introduce the Cyclotomic Coset Problem (CCP), a cyclotomic generalization of DCP that preserves an exact hidden-subgroup structure.  Let $\zeta_p$ be a primitive $p$-th root of unity, let $\pi=\zeta_p-1$, and write $q=p^t$ and $L=t(p-1)$.
We work over $R_q=\mathbb Z_q[\zeta_p] \cong \mathbb Z[\zeta_p]/(\pi^L)$, 
where the isomorphism follows from the total ramification identity $(p)=(\pi)^{p-1}$.
We exploit the resulting $\pi$-adic ideal chain to construct a quantum sieve that 
successively reduces phase states modulo $\pi^{L},\pi^{L-1},\ldots,\pi$. For every 
fixed prime $p$ and modulus $q=p^t$, our algorithm solves the CCP in time and 
sample complexity $2^{O_p(\log n\log q)}$, using polynomial quantum space. 
The sieve also applies to uniform EDCP and Gaussian $S\ket{\LWE}$, yielding 
quasi-polynomial time algorithms for all the above problems when $q=\poly(n)$. 
This extends the power-of-two EDCP sieve of Bai et al. (CRYPTO 2025) to a cyclotomic setting.

However, we emphasize that our result does not, by itself, yield a quasi-polynomial-time algorithm for standard LWE, because the currently known reduction produces only a limited number of approximate CCP states.
%
%The Learning With Errors (LWE) problem is a fundamental assumption in post-quantum cryptography. Combining 
%Regev’s quantum reductions for lattice problems (FOCS 2002) with his worst-case hardness proof 
%for LWE (STOC 2005) yields a reduction from LWE to the Dihedral Coset Problem. Later, Brakerski et al. (PKC'18) introduced the Extrapolated Dihedral Coset Problem (EDCP), proving its equivalence to LWE. However, unlike DCP, EDCP no longer admits a coset structure. This limits the applicability of techniques inspired by the solving of hidden subgroup problems.
%
%In this work, we introduce the Cyclotomic Coset Problem (CCP), a cyclotomic generalization of DCP that preserves an exact hidden-subgroup structure. Let $\zeta_p$ be a primitive $p$-th root of unity and $\pi=1-\zetap$. Since $p$ is totally ramified in $\QQ(\zetap)$, the element $\pi$ is the generator of the unique prime ideal above $p$, and $(p)=(\pi)^{p-1}$. Exploiting the resulting $\pi$-adic ideal chain, we develop a quantum sieve that progressively reduces phase states modulo $\pi^{L/},\pi^{L-1},\ldots,\pi$. For every fixed prime $p$ and modulus $q=p^t$, our algorithm solves the CCP in time and sample complexity $2^{O_p(\log n\log q)}$, using polynomial quantum space. This provides a cyclotomic analogue of the power-of-two EDCP sieve of Bai et al. (CRYPTO 2025). We emphasize that our result does not, by itself, yield a quasi-polynomial-time algorithm for standard LWE, because the currently known reduction produces only a limited number of approximate CCP states.\TODO{modifier }
\end{abstract}


\section{Introduction}

The Learning With Errors problem (LWE), introduced by Regev in~\cite{regev_lattices_2005}, is a cornerstone of post-quantum cryptography. In its search form, LWE asks to recover a secret vector $\s\in\mathbb{Z}_q^n$ from noisy linear equations

$$b_i=\langle \mathbf{a}_i,\s\rangle+e_i \bmod q,$$
where the vectors $\mathbf{a}_i$ are uniform in $\mathbb{Z}_q^n$ and the errors $e_i$ are sampled from a prescribed narrow distribution. The importance of LWE stems both from its versatility in cryptographic constructions and from quantum worst-case-to-average-case reductions relating it to fundamental lattice problems~\cite{regev_lattices_2005,Regev09}.

Quantum coset problems provide another perspective on the complexity of lattice problems. In earlier work, Regev established a connection between unique shortest-vector problems and the Dihedral Hidden Subgroup Problem~\cite{Regev02}. Its coset-state formulation, known as the Dihedral Coset Problem (DCP), asks to recover a 
secret $\s\in\mathbb{Z}_q^n$ from states of the form

$$\frac{1}{\sqrt 2} \left(\ket{0}\ket{\x}+ \ket{1}\ket{\x+\s} \right),$$
where $\x\in\mathbb{Z}_q^n$ is uniformly random. DCP is particularly appealing because its input states are exact coset states of order-two subgroups \mbox{of a generalized dihedral group.}

Childs and van Dam considered in \cite{childs_quantum_2007} the generalized hidden shift problem. It consists in finding a hidden shift $s \in \ZZ_q$ given quantum samples of the form 
\[\frac{1}{\sqrt{M}}\sum_{j = 0}^{M-1} \ket{j}\ket{x_i + j s}\]
for randomly chosen $x_i \in \ZZ_q$  and an integer $M$. They provided an efficient quantum algorithm based on "pretty good measurement" when $M$ is large with respect to $q$ ($M=q^\epsilon$ for a constant $\epsilon$).

Brakerski, Kirshanova, Stehl\'e, and Wen subsequently introduced the Extrapolated Dihedral Coset Problem (EDCP)~\cite{brakerski_learning_2018}. In EDCP, \mbox{the DCP state is replaced by}

$$\sum_j f(j)\ket{j}\ket{\x+j\s}$$
for a prescribed amplitude function $f$. They proved that suitable Gaussian and uniform variants of EDCP are equivalent to LWE under quantum polynomial-time reductions, up to parameter losses. This equivalence provides a useful quantum formulation of the computational content of LWE. However, except in particular cases, the set of translations

$$ \{(j\s,j):j\in\operatorname{supp}(f)\}$$
does not form a subgroup. General EDCP states therefore do not retain the \mbox{exact hidden-subgroup structure of DCP.}

Closely related quantum formulations arise by encoding the LWE error distribution in the amplitudes of a quantum state. In the $S\ket{\mathrm{LWE}}$ problem introduced by Chen, Liu, and Zhandry~\cite{CLZ22}, one is given uniformly random vectors $\mathbf{a}_i\in\mathbb{Z}_q^n$, together with quantum states of the following form

$$\sum_{e\in\mathbb{Z}_q} f(e) \ket{\langle\mathbf{a}_i,\s\rangle+e\bmod q},$$
and the objective is again to recover $\s$. Applying a quantum Fourier transform moves the secret into a phase and produces states closely related \mbox{to phase formulations of EDCP.}

\paragraph{Quantum sieving.}
DCP admits subexponential-time quantum algorithms due to Kuperberg~\cite{kuperberg_subexponential-time_2005,Kup13,doliskani_efficient_2020} and Regev~\cite{regev_subexponential_2004}. At a high level, these algorithms first transform coset states into phase-states of the form 
$1/\sqrt 2 (\ket{0}+ \omega_q^{\langle\y,\s\rangle}\ket{1}),$
with a label $\y$ which is uniform and known. They then repeatedly combine phase-states whose labels agree on selected blocks, thereby producing new states with increasingly divisible labels. The process eventually isolates enough \mbox{information to recover the secret.}

Recently, Bai, Jangir, Kirshanova, Ngo, and Youmans developed a quasi-polynomial-time quantum algorithm for EDCP over power-of-two moduli~\cite{bai_quasi-polynomial_2025} inspired by the "Simon-meets-Kuperberg" algorithm of Bonnetain and Naya-Plasencia \cite{bonnetain_hidden_2018}. Their algorithm uses a quasi-polynomial number of quantum samples and polynomial quantum space. It also yields a quasi-polynomial-time algorithm for Gaussian-$S\ket{\mathrm{LWE}}$ over power-of-two moduli. The algorithm exploits the following sequence of ideals:
$$\mathbb{Z}_{2^t} \supset 2\mathbb{Z}_{2^t} \supset\cdots\supset 2^{t-1}\mathbb{Z}_{2^t},$$
combining binary phase states to increase the regularity of their labels one level at a time. Chailloux and Hermouet~\cite{CH25} recovered the same result using a reduction from $\mathrm{S}\ket{\LWE}$ to the Inhomogeneous Short Integer Solution problem (ISIS). Although using a different technique, they encountered the same bottleneck as in~\cite{bai_quasi-polynomial_2025} when trying to generalize the \mbox{result to other prime-power moduli.}

Hidden-shift problems such as DCP can be viewed as hidden-subgroup problems. Thus, in parallel with \cite{bai_quasi-polynomial_2025}, Imran and Ivanyos \cite{Imran:2023tft} generalized the “Simon-meets-Kuperberg” algorithm \cite{bonnetain_hidden_2018} to a particular class of groups, namely nilpotent groups. These are groups for which there exists a subgroup chain allowing us to use the same sieving techniques as in the case where we have the 2-adic chain mentioned \mbox{in the paragraph just above.}

However, Imran and Ivanyos’ result cannot be directly applied to EDCP, since it does not have the structure of a hidden subgroup; nor can it be applied to a DCP with parameters $(n, q)$ for $q$ that is not a power of two, as stated in the following exercise in Rotman \cite[Exercise~5.41]{rotman1999introduction}:
\begin{center}
  \emph{The dihedral group of order $2q$ is nilpotent if and only if $q$ is a power of two.}
\end{center}
One might therefore ask whether there exists a specific nilpotent group that can be related to EDCP and DCP for any prime power $q$. This is what motivates the introduction of a new coset problem, which we have \mbox{named the Cyclotomic Coset Problem.}
 
We emphasize that our algorithms are not a direct application of Imran and Ivanyos~\cite{Imran:2023tft}. The exact algorithm of Imran and Ivanyos assumes access to a unitary $U$ implementing state preparation and to its inverse $U^{-1}$. In our model, the input consists only of independent quantum samples: neither the preparation circuit nor its inverse is available. Their exact theorem therefore does not apply directly. Moreover, the nilpotency class $t(p-1)$ of our cyclotomic group $G_{n,q,p}:=R_q^n\rtimes\mathbb{F}_p$ need not be constant. We instead use their zero-sum algorithm and describe the quantum state transformations explicitly, accounting for \mbox{sample consumption and working space.}

\subsection{Cyclotomic Coset}

Let $p$ be a prime, let $q=p^t$, and let $\zeta_p$ be a primitive $p$-th root of unity. We consider the cyclotomic ring $R_q:=\mathbb{Z}_q[\zeta_p]$. Multiplication by $\zeta_p$ defines an action of the additive group $\mathbb F_p$ on $R_q$. This allows us to define the semidirect product $G_{n,q,p}:=R_q^n\rtimes\mathbb{F}_p$ with multiplication
$$(\x,j)(\y,k) = \bigl( \x+\zeta_p^j\y, j+k \bigr).$$

Consequently, for every $\s\in R_q^n$, the set
$$
  H_{\s}
  :=
  \{
    (\lambda_j\s,j):
    j\in\mathbb{F}_p
  \}, \qquad \text{with } \lambda_j = \sum_{i = 0}^{j-1} \zetap^i,
$$
is a subgroup of $G_{n, q,p}$. Equivalently, it is the cyclic subgroup generated by $(\s,1)$. The subgroup order follows from the identity

$$\lambda_p = 1+\zeta_p+\cdots+\zeta_p^{p-1} = 0$$
shows that $H_{\s}$ has order $p$.
A left coset of $H_{\s}$ gives rise to the quantum state
$$\ket{\phi_{\x,\s}} =
  \frac{1}{\sqrt p} \sum_{j\in\mathbb{F}_p} \ket{j} \ket{\x+\lambda_j\s},$$
where $\x\in R_q^n$ is uniformly random. We define the Cyclotomic Coset Problem (CCP) as the 
problem of recovering $\s$ from \mbox{independent states of this form.}

When $p=2$, we have
$$\zeta_2=-1,\qquad \lambda_0=0,\qquad \lambda_1=1,$$
and $R_q=\mathbb{Z}_q$. Hence \mbox{CCP specializes exactly to DCP.}

We show that $G_{n,q,p}$ is a nilpotent group, in \Cref{le:cyclotomic_grp_is_nilpotent} and 
the associated hidden subgroup problem is close \mbox{to a generalized hidden-shift problem.} 

Note that our version of (CCP) assumes that the secret $\s$ is contained in $\ZZ_q^n$, a restricted subset of $R_q^n$. This choice is justified by the introduction of a cyclotomic version of LWE (CyLWE), which \mbox{also requires a restricted secret.}

\subsection{Organization and contributions}

Our work makes several contributions towards the study of coset problems over cyclotomic rings and their quantum algorithms.
\begin{itemize}

\item We first introduce new computational problems over cyclotomic rings, namely CyLWE and CCP, which can be viewed as analogues of LWE and DCP, respectively. We establish equivalences between these problems when $p$ is constant.

\item We define a quantum Fourier transform over cyclotomic rings, which allows us to give a rigorous definition of phase-state samples over such rings, in \mbox{analogy with Kuperberg's phase states.}

\item We construct a sieving algorithm for CCP that exploits these new phase-state samples. Our algorithm can be viewed as a specific instance of the framework introduced by Imran and Ivanyos~\cite{Imran:2023tft}.

\item We show that, for a certain range of parameters, our algorithm solves uniform EDCP and $S\ket{\LWE}$ in quasi-polynomial time, given a quasi-polynomial number of samples. This generalizes the results of~\cite{bai_quasi-polynomial_2025} from power-of-two \mbox{moduli to a broader setting.}

\end{itemize}

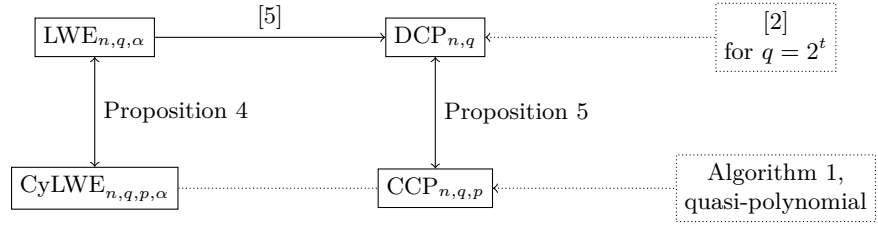
\begin{figure}[H]
    \centering
    \begin{tikzpicture}[
    node distance=2 and 4.5, on grid,
    algo/.style={draw,},
    refe/.style={draw, align=center, densely dotted},
    label/.style={midway}
]

  \node[algo] (LWE) {$\LWE_{n, q, \alpha}$};
  \node[algo, right =of LWE] (EDCP) {$\DCP_{n, q}$};
  \node[refe, right =of EDCP] (BJK) {\cite{bai_quasi-polynomial_2025} \\ for $q = 2^t$};
  \node[algo, below =of LWE] (CYLWE) {$\CyLWE_{n, q, p,\alpha}$};
  \node[algo, right =of CYLWE] (CCP) {$\CCP_{n, q, p}$};
  \node[refe, right =of CCP] (QP) { \Cref{al:CCP_solver}, \\ quasi-polynomial};

  \draw[->] (LWE) -- node[label, above] {\cite{brakerski_learning_2018}} (EDCP);
  \draw[<-, refe] (EDCP) -- (BJK);
  \draw[-, refe] (CYLWE) --node[label, above] {}  (CCP);
  \draw[<-, refe] (CCP) -- (QP);
  \draw[<->] (LWE) --node[label, right] {Proposition \ref{th_LWE_to_CyLWE}} (CYLWE);
  \draw[<->] (EDCP) --node[label, right] {Proposition \ref{pro:DCP_and_CCP}} (CCP);
\end{tikzpicture} 
\caption{Overview representing the different contributions of the paper. Note that we are 
considering an $n$-dimensional version of DCP. Moreover,~\cite{bai_quasi-polynomial_2025} 
shows that there exists a reduction from $\operatorname{EDCP}_{n, q, M}$ to $\DCP_{n, q}$.}
\end{figure}

\section{Preliminaries}

\paragraph{Asymptotic notations.} 
We define the notations $O(\cdot), \Omega(\cdot), \poly(\cdot)$ in the standard way, with respect to the dimension $n$ of the lattice and the security parameter $\kappa$. We say a function $f: \mathbb{N} \to (0, 1]$ is negligible if for all positive polynomials $p(\cdot)$ there exists an integer $N$ such that, $\forall n > N, f(n) < 1 / p(n)$. We say a positive function is quasipolynomial if it is upper bounded by $2^{O(\log^c n)}$ \mbox{for some constant $c\geq1$.}

\subsection{Cyclotomic ring}

We introduce the notation and definitions associated with the cyclotomic ring used throughout the paper. Let $p$ be a prime number and let $q = p^t$ with $t\geq1$. We denote by $\zetap$ a primitive $p$-th root of unity, and define the cyclotomic field $K = \QQ[\zetap]$ as an extension of $\QQ$ generated by $\zetap$. Equivalently, the cyclotomic field can be viewed as the quotient $\QQ[X]/(\Phi_p (X))$ where $\Phi_p (X) := 1 + X + \dots + X^{p-1}$ is the $p$-th cyclotomic polynomial. Hence we denote $R := \ZZ [\zetap]$ the ring of integers of $K$. We write $R_q = R/q R$. We can view $R_q$ as $\ZZ_q [\zetap] = \ZZ_q [X]/(\Phi_p (X))$, a cyclotomic extension of $\ZZ_q$. This allows us to represent elements of $\cycq$ by their coefficients with respect to the basis $(1, \zetap, \dots,\zetap^{p-2})$,
\[\iota: \ZZ_q^{p-1} \to \cycq, \qquad (x_0, \dots, x_{p-2}) \mapsto \sum_{i = 0}^{p-2} x_i \zetap^i\]
and we can state that $\iota$ is an isomorphism of $\ZZ_q$-modules.

Using the same arguments as in the previous paragraph, one can show that $R := \ZZ [\zetap]$ and $\ZZ^{p-1}$ are isomorphic as $\ZZ$-modules. This directly induces a notion of distance on the cyclotomic ring. Indeed, for $e \in \ZZ [\zetap]$, we use the slight abuse of notation $\norm{e}_2$ to denote the $\ell_2$-norm of its representation in $\ZZ^{p-1}$ with respect to the basis $(1, \zetap, \dots, \zetap^{p-2})$.

Finally, we denote by $\inner{\cdot}{\cdot}$ the symmetric bilinear form on $\cycq^n$ such that, for ${\mathbf{x}} = (x_1, \dots, x_n)$ and ${\mathbf{y}} = (y_1, \dots, y_n)$ in $\cycq^n$, we have

\[\inner{{\mathbf{x}}}{{\mathbf{y}}}:=\sum_{i = 1}^{n} x_i y_i.\]

\paragraph{$\pi$-decomposition.} The element $\pi = \zetap - 1$ is a prime element of $R$, and there exists $u_0 \in R^\times$ invertible such that
\[p = u_0 \pi^{p-1},\]
and therefore $q = u_0^t \pi^L$ with $L = t (p-1)$, we say that $p$ is ramified. Consequently, the two quotient rings are canonically isomorphic:
\[R_q = R/ q R \cong R/\pi^L R.\]
For this reason, we will sometimes use the abusive notation $R_L$ to talk about $R_{\pi^L}$. We have $|\ring{k}|=p^k$: multiplication by $\pi^{k-1}$ identifies $R/\pi R$ with $\pi^{k-1}R/\pi^kR$, so each successive quotient has $p$ elements. This will allow us to consider intermediate quotient rings that would be difficult to represent without the $\pi$-decomposition:
\[R_q = R_L \to R_{L-1} \to \dots \to R_1 = \FF_p.\]

\paragraph{Field trace.}The field trace corresponds to the linear map $ \Tr : K \to \QQ$.
More precisely, suppose $K/\QQ$ is a Galois extension and let $x$ be an element of $K$. Then the trace of $x$ is the sum of all Galois conjugates of $x$:
\begin{equation*}
  \Tr(x) = \sum_{\sigma \in \operatorname{Gal}(K/\QQ)} \sigma(x).
\end{equation*}

Let $K = \QQ(\zetap)$ be a cyclotomic extension of $\QQ$. The trace is the linear map from $K$ to $\QQ$ such that
\begin{equation*}
  \Tr(1) = p-1, \qquad \Tr(\zetap^k) = -1 \quad \text{for all } 1 \leq k \leq p-1 .
\end{equation*}
This follows from the fact that $\operatorname{Gal}(K/\QQ) = \{\sigma_u \in \operatorname{Aut}(K) : \sigma_u(\zetap) = \zetap^u,\ u \in \FF_p^\times\}$. We immediately see that $\Tr(1) = p-1$. For $1 \leq k \leq p-1$, we obtain
\begin{equation*}
  \Tr(\zetap^k) = \sum_{u \in \FF_p^\times} \sigma_u(\zetap^k) = \zetap + \cdots + \zetap^{p-1} = -1.
\end{equation*}

\subsection{Lattices and cyclotomic lattices}
\label{sec:cyc lattices}

\paragraph{Discrete gaussian distribution.} For any $r > 0$ and $p >2$, we define the Gaussian function $\rho_r(x) := \exp\!\left(-\pi \norm{x}_2^2 / r^2\right)$ for $x \in \RR^{n}$. For $x \in \cycq^n$, we also use the abusive notation $\rho_r (x)$ to denote the quantity $\rho_r (\iota(x))$ with $\iota(x) \in \RR^{n (p-1)}$. For an $n$-dimensional lattice $\Lambda$, we write $D_{\Lambda, r}$ for
the discrete Gaussian distribution over $\Lambda$ with density proportional to $\rho_r(\cdot)$. For the coefficient lattice $\mathbb Z^{n(p-1)}$, the Gaussian weight factors over coordinates, so the resulting distribution is a product of independent one-dimensional discrete Gaussians~\cite{aggarwal_note_2013}. This factorization is not asserted for arbitrary lattices. We also extend the definition of Gaussian distribution to $R^n$ by considering its injection in $\RR^{n(p-1)}$ related to the basis $(1, \zetap, \dots, \zetap^{p-2})$.

The tails of a Gaussian distribution are negligible \mbox{compared to its central mass.}

\begin{lemma}[\cite{banaszczyk_new_1993}, Lemma 1.5]
\label{gaussian}
  For any $n$-dimensional lattice $\Lambda$ and $r>0$,
  \[\rho_r\!\left(\Lambda \setminus B_n (0,\, \sqrt{n}\, r)\right)
    < 2^{-\Omega(n)}\, \rho_r(\Lambda).\]
\end{lemma}

\noindent
Note that this result holds for cyclotomic sublattices of $R^n$ \mbox{using their coefficient embedding in $\RR^{n(p-1)}$.}

\iffalse
\paragraph{$q$-ary cyclotomic lattices.} Let $A \in \mathcal{U}(\cycq^{m \times n})$ be a matrix and
  define the $m$-dimensional $q$-ary lattice over cyclotomic integers
  \[
    \Lambda_q(A) = \{y \in  R^m \mid y= Ax \mod q, x \in R_q^n\}.
  \]

  \begin{restatable}[Inspired by \cite{wen_thesis_2018} Lemma 2.18]{lemma}{minimaldistance}
  \label{le_minimal_distance}
  Let $A \in \mathcal{U}(\cycq^{m \times n})$ be a random matrix with $m \geq n$. The minimum distance of $\Lambda_q(A)$ with respect to norm 2 related to the basis $(1, \zetap, \dots, \zetap^{p-2})$ defined in the preliminaries satisfies,
  \[
    \lambda_1(\Lambda_q(A))
    \geq \min (\frac{1}{6} \left(\frac{p-1 }{p}\right)^{1/2} q^{(m-n)/m},q)
  \]
  with probability $1 - 2^{-(p-1)m}$.
\end{restatable}

The proof is given in \Cref{sec:missing proofs}.

\yixin{j'ai vérifié la preuve de chatgpt, il est correct. Je l'ai modifié un peu et l'a mis dans l'annexe.}

\fi

\subsection{The cyclotomic group}

Let $p$ be a prime number and $q = p^t$ a power of $p$. We call $\zetap$ a primitive $p$-th root of unity and denote the cyclotomic ring by $R_q := \ZZ_q [\zetap]$, and we fix a dimension $n\geq1$. We are interested here in the group
\[G_{n, q, p} = R^n_q \rtimes \FF_p. \]
This is the group $(R^n_q \rtimes \FF_p, \star)$ where the product $\star$ satisfies:
\[(\x_1, j_1) \star (\x_2, j_2) := (\x_1 + \zetap^{j_1} \x_2, j_1 + j_2),  \text{ for } \x_1, \x_2 \in R_q^n \text{ and } j_1, j_2 \in \ZZ_p\]
One can verify that the internal composition law $\star$ indeed defines \mbox{a group structure on $G_{n,q,p}$.}

\paragraph{Nilpotent group.} Let $G$ be a group. Let $A$ and $B$ be two subgroups of $G$; we denote by $[A, B]$ the subgroup generated by the commutator $[a, b] := a b a^{-1} b^{-1}$ with $a$ in $A$ and $b$ in $B$. We call the lower central series the sequence of subgroups $C^k (G)$ defined for every positive integer $k$ by $C^1 (G)=G$ and $C^{k+1} (G) = [G, C^k (G)]$. We say that $G$ is \emph{nilpotent} if there exists an integer $k$ such that $C^k (G) = \{e\}$ (the trivial group). Moreover, the nilpotency class of $G$ is the smallest integer $k$ such that $C^{k+1} (G) = \{e\}$. The lemma below allows us to state that the cyclotomic \mbox{group is a nilpotent group.}

\begin{lemma}\label{le:cyclotomic_grp_is_nilpotent}
  Let $p$ be a prime, $q = p^t$ with $t\geq1$, and let $n\geq1$ be a dimension. The cyclotomic group $G_{n, q, p}$ is a nilpotent group, of nilpotency class $t (p-1)$.
\end{lemma}

\begin{proof}
  Let $G = R_q^n \rtimes \FF_p$ and $L=t(p-1)$. We shall show by induction that, for every integer $k \geq 2$,
  \(C^k(G)=\pi^{k-1} R_q^n \rtimes \{0\}.\)
  In the case $k=2$, we have for $(\x_1,j_1),(\x_2,j_2)\in G$:
  \[ [(\x_1,j_1),(\x_2,j_2)] = \left((1-\zetap^{j_2})\x_1+(\zetap^{j_1}-1)\x_2,0\right).\]
  Since $1-\zetap^{j}=-\pi(1+\zetap+\cdots+\zetap^{j-1})$, we obtain \(C^2(G)\subseteq \pi R_q^n\rtimes\{0\}.\)
  Conversely, for every $\x\in R_q^n$,
  \begin{equation}\label{eq:nilpotence}
     [(0,1),(\x,0)]=(\pi \x,0),
  \end{equation}
  hence \(C^2(G)=\pi R_q^n\rtimes\{0\}\). Now suppose that, for some $k\geq2$,
  \(C^k(G)=\pi^{k-1} R_q^n\rtimes\{0\}\).
  Commuting an element of this subgroup with any element of $G$ multiplies its vector by $\zetap^j-1$, so $C^{k+1}(G)\subseteq\pi^k R_q^n\rtimes\{0\}$. Equation~\eqref{eq:nilpotence}, applied to $\x\in\pi^{k-1}R_q^n$, proves the reverse inclusion.

  Finally, $(q)=(\pi^L)$ gives $\pi^L R_q^n=0$ and $\pi^{L-1}R_q^n\neq0$. Therefore $C^{L+1}(G)=\{e\}$ and, for $L\geq2$, $C^L(G)\neq\{e\}$. If $L=1$, $G$ is nontrivial and abelian. In both cases its nilpotency class is exactly $L=t(p-1)$. \qed
\end{proof}

\subsection{Computational problems}

\paragraph{Learning with errors problems.}We introduce a variable $\kappa$ to relate all the parameters involved in the definition below.  Indeed, $n$, $q$, $p$ are functions of $\kappa$, but we omit the variable $\kappa$ for clarity.

\begin{definition}[Search LWE, \cite{regev_lattices_2005}]
  Let $p$ be a prime, $q = p^t$, $n$ be a positive integer, and $\alpha > 0$.
  Fix a secret $\s \in \ZZ_q^n$.
  Given $\ell$ samples of the form
  \[\left(\a,\; b = \inner{\a}{\s}  + e \bmod q\right),\]
  with $\a \leftarrow \calU(\ZZ_q^n)$ and $e \leftarrow D_{\ZZ,\, \alpha q}$, denoted by $\LWE^\ell_{n,q,\alpha}$,the {search LWE problem},
  asks to recover $\s \in \ZZ_q^n$.
\end{definition}

We now introduce a version of $\LWE$ on the cyclotomic ring called CyLWE. 
The rational-secret restriction and the coefficient error distribution distinguish this problem from the usual Ring-LWE (RLWE) setting discussed in \cite{wen_module_2026}. For fixed $p$, the ring degree is the constant $p-1$, and $n$ is an independent vector dimension. In common RLWE parameter families, the ring degree \mbox{grows with the security parameter.} 

\begin{definition}[Cyclotomic LWE]
  Let $p$ be a prime, $q = p^t$, $n$ be a positive integer, and $\alpha > 0$.
  Fix a secret $\s \in \ZZ_q^n$.
  Consider $\ell$ samples of the form
  \[\left(\a,\; b = \inner{\a}{\s}  + e \bmod q\right),\]
  with $\a \leftarrow \calU(\cycq^n)$ and $e \leftarrow D_{\cyc,\, \alpha q}$. The {cyclotomic LWE} problem, denoted by $\CyLWE^\ell_{n,q,p,\alpha}$, asks to recover ${\mathbf{s}} \in \ZZ_q^n$, given $\ell$ samples.
\end{definition}

\paragraph{Hidden subgroup problem.}

\begin{definition}[Hidden Subgroup Problem]
  Let $G$ be a finite group and $X$ a finite set. Fix $H$ a hidden subgroup of $G$. Given a function $f : G \to X$ that hides $H$ (i.e., $f(g)=f(g\prime)$ if and only if $gH=g\prime H$), provided by means of an oracle using $O(\log|G| + \log|X|)$ bits, the Hidden Subgroup Problem (HSP) asks to recover a generating set of $H$.
\end{definition}

In practice, the oracle used in the HSP is assumed to be the existence of a quantum gate $U_f$ such that $\ket{g}\ket{0} \mapsto \ket{g}\ket{f(g)}$. A quantum use of this gate is as follows. First, this consists of generating a quantum superposition of the elements of $G$ and computing $f$ in a second register by means of $U_f$.
\[ \frac{1}{\sqrt{|G|}} \sum_{g \in G} \ket{g}\ket{0} \to \frac{1}{\sqrt{|G|}} \sum_{g \in G} \ket{g} \ket{f(g)}.\]
The second register is then measured. By the hiding promise, the preimage of an output $x \in X$ of $f$ can be viewed as a left coset of $H$, i.e., there exists $g_0 \in G$ such that $f^{-1}\{x\} = g_0 H$. The resulting state is therefore
\[\frac{1}{\sqrt{|H|}} \sum_{h \in H} \ket{g_0 h}.\]
Thus, a quantum instance that is solved in place of the classical instance of the \mbox{HSP is stated as follows.}

\begin{definition}[Coset Problem]
  Let $G$ be a group and $H$ a hidden subgroup. Consider samples of the form
  \[\frac{1}{\sqrt{|H|}} \sum_{h \in H} \ket{g_i h},\]
  with the $g_i$ drawn independently and uniformly from $G$. Given a certain number of samples from $G$, the Coset Problem asks to recover a subset generating $H$.
\end{definition}

For hidden subgroups generated by $(\s,1)$, this gives the dihedral coset problem over $\ZZ_q^n\rtimes_{-1}\ZZ_2$, and the cyclotomic coset problem over $G_{n,q,p}$, respectively.

\paragraph{Coset problems.} We now introduce two quantum hidden-shift problems over finite rings, namely the Dihedral Coset Problem and its cyclotomic generalization. Initially, the dihedral coset problem was stated over $\ZZ_q$, but \cite{doliskani_efficient_2020} discusses the $\DCP$ over $\ZZ_q^n$. We have therefore chosen to state it directly over $\ZZ_q^n$ too.

\begin{definition}[Dihedral Coset Problem]
  Let $p$ be a prime, $q=p^t$ with $t\geq1$, and $n\geq1$.
  Fix a secret $\s \in \ZZ_q^n$.
  Consider quantum samples $\{\ket{\phi_k}\}_{k=0}^{\ell-1}$ of the form
  \[
    \ket{\phi_k}
    = \frac{1}{\sqrt{2}}\left(\ket{0}\ket{\x_k} + \ket{1}\ket{\x_k + \s}\right),
  \]
where the $\x_k$ are independent uniform elements of $\ZZ_q^n$. The {Dihedral Coset Problem} over $\ZZ_q^n$, denoted by $\DCP^\ell_{n,q,p}$, asks to recover $\s \in \ZZ_q^n$, given $\ell$ samples.
\end{definition}

\begin{definition}[Cyclotomic Coset Problem]
  Let $p$ be a prime, $q=p^t$ with $t\geq1$, and $n\geq1$.
  Fix a secret $\s \in \ZZ_q^n$. We identify $\s \in \ZZ_q^n$ with its coordinatewise constant embedding in $\cycq^n$.
  Consider quantum samples $\{\ket{\phi_k}\}_{k=0}^{\ell-1}$ of the form
  \[
    \ket{\phi_k}
    = \frac{1}{\sqrt{p}} \sum_{j\in \FF_p} \ket{j}\,
      \ket{\x_k + \lambda_j \s \bmod q},
  \]
where $\lambda_0=0$, $\lambda_j=\sum_{i=0}^{j-1}\zetap^i$ for $1\leq j<p$, and the $\x_k$ are independent uniform elements of $\cycq^n$. The {Cyclotomic Coset Problem}, denoted by $\CCP^\ell_{n,q,p}$, asks to recover $\s \in \ZZ_q^n$, given $\ell$ samples.
\end{definition}

These problems are closely related to the EDCP introduced in
\cite{brakerski_learning_2018}; for $p=2$, CCP is DCP and hence uniform EDCP with two branches. 
Furthermore, it is clear that these two problems are specific variants of the Coset Problem. For DCP, the ambient group is the generalized dihedral group, while CCP uses $G_{n,q,p}$. Indeed, $\lambda_{j+k}=\lambda_j+\zetap^j\lambda_k$ (indices modulo $p$), so $(\s,1)^j=(\lambda_j\s,j)$ and $H_\s=\langle(\s,1)\rangle$ has order $p$. The left coset $(\x,0)H_\s$ gives exactly the stated \mbox{CCP state after swapping registers.}

\subsection{Quantum Fourier Transform}

We briefly recall the definition of the Quantum Fourier Transform and its efficient \mbox{implementation as a quantum circuit.}

\begin{definition}[Quantum Fourier Transform]
Let $N\geq1$ be an integer and let $\mathcal{H}_N$ be the $N$-dimensional Hilbert space spanned by the computational basis
${\ket{0},\ldots,\ket{N-1}}$. The \emph{Quantum Fourier Transform} (QFT) over $\mathbb{Z}_N$ is the unitary operator $\operatorname{QFT}_N$ defined on the computational basis states by
\[
\operatorname{QFT}_N (\ket{x}) =
\frac{1}{\sqrt{N}} \sum_{y=0}^{N-1} e^{2\pi ixy/N}\ket{y},
\qquad \forall x\in\{0,\ldots,N-1\},
\]
and extended \mbox{to arbitrary states by linearity.}
\end{definition}

\begin{lemma}[\cite{dewolf2026quantumcomputinglecturenotes}, Implementation of the QFT]
Let $N\geq2$ and $0<\varepsilon<1$. The quantum Fourier transform $\operatorname{QFT}_N$ admits a circuit approximation with operator-norm error at most $\varepsilon$ using $\operatorname{poly}(\log N,\log(1/\varepsilon))$ elementary gates from a fixed universal gate set. When $N$ is a power of a fixed prime, its ideal radix expansion uses $O(\log^2N)$ \mbox{constant-dimensional Fourier and controlled-phase gates.}
\end{lemma}

\section{Quantum Fourier Transform over cyclotomic rings}

In this section, we define the quantum Fourier transform (QFT) on $\ring{L} = R/\pi^L R$ and provide an efficient construction of it. We then introduce phase-states, which are quantum states directly obtained from CCP samples using the QFT and which will be useful \mbox{for designing a sieving algorithm.}

\subsection{The duality in $R_L$}

In this part, we provide an explicit description of the dual of $\ring{L}$ and its characters, using the field trace. The characters are indexed by the elements of $\ring{L}$ \mbox{and satisfy several useful properties.}

\begin{lemma}\label{le:bilinear_trace}
  Fix $K = \QQ(\zetap)$ and let $L$ be a positive integer. The map
  \begin{equation*}
    B_L : \ring{L} \times \ring{L} \to \QQ/\ZZ, \qquad B_L(x,y) := \Tr\left(\frac{xy}{p\pi^{L-1}}\right) \bmod \ZZ
  \end{equation*}
  is a non-degenerate symmetric $\ZZ$-bilinear form.
\end{lemma}

\begin{proof}
  Let $(\zetap^i)_{0 \leq i \leq p-2}$ be the canonical basis of $R$. We can show that the family $(e_j)_{0 \leq j \leq p-2}$ of $R$ defined by
  \begin{equation*}
    e_j := \frac{\zetap^{p-j} - \zetap}{\pi} = \zetap \sum_{k=0}^{p-j-2} \zetap^k, \qquad 0 \leq j \leq p-2,
  \end{equation*}
  is also a basis of $R$: $e_0=-1$, and $e_j-e_{j+1}=\zetap^{p-j-1}$ for $0\leq j<p-2$, so these elements generate the basis $1,\zetap,\ldots,\zetap^{p-2}$. For $0 \leq i, j \leq p-2$ we have $\Tr\!\left(\frac{\pi}{p}\zetap^i e_j\right) = \delta_{i,j}$. In other words, $((\pi/p)e_j)_{0 \leq j \leq p-2}$ is the trace-dual basis of $(\zetap^i)_{0 \leq i \leq p-2}$. This allows us to show, for $x, y \in R$, that
  \begin{equation}\label{eq:non-degenerate}
    \Tr\left(\frac{xy}{p\pi^{L-1}}\right) \in \ZZ \text{ for all } y \in R \iff x \in \pi^L R,
  \end{equation}
  because the trace-dual lattice is $R^\vee=(\pi/p)R$, and the left-hand condition is equivalent to $x/(p\pi^{L-1})\in R^\vee$.

  The map $B_L$ is well defined: changing either lift by an element of $\pi^L R$ changes the trace by an element of $\Tr((\pi/p)R)\subseteq\ZZ$. It is symmetric and $\ZZ$-bilinear, and non-degeneracy follows from~\eqref{eq:non-degenerate}.
  \qed
\end{proof}

Thanks to this lemma, one can easily define the dual of $R_L$, by indexing the characters by the elements of $R_L$.

\begin{proposition}[Duality]\label{pr:duality}
  The map below is an isomorphism:
  \begin{equation*}
    \Phi_L : R/\pi^LR \to \operatorname{Hom}\big((R/\pi^LR, +),\, \CC^\times\big), \qquad y \mapsto \chi_y,
  \end{equation*}
  where $\chi_y : x \mapsto \exp\!\left(2i\pi \Tr\left(\frac{xy}{p\pi^{L-1}}\right)\right)$. For $y \in R/\pi^LR$, we call \emph{character} of $(\ring{L}, +)$, the map $\chi_y : \ring{L} \to \CC^\times$.
\end{proposition}

\begin{proof}
  Since the field trace is linear, $\Phi_L$ is a group homomorphism. It remains to show that $\Phi_L$ is a bijection. Because $B_L$ from \Cref{le:bilinear_trace} is non-degenerate, $\Phi_L$ is injective. Moreover, $(\ring{L}, +)$ is a finite group, so it has the same cardinality as its dual. In other words,
  \begin{equation*}
    |\ring{L}| = |\operatorname{Hom}((\ring{L}, +), \CC^\times)|.
  \end{equation*}
  Hence, $\Phi_L$ is an injection between groups of the same cardinality, \mbox{so it is an isomorphism.}\qed
\end{proof}

Several useful properties of \mbox{characters follow from this definition.}

\begin{proposition}[Properties of the Characters]\label{pr:characters}
  Let $x, y, a, b$ be elements of $\ring{L}$, let $\lambda \in R$, and let $\sigma_u \in \operatorname{Gal}(K/\QQ)$ with $u \in \FF_p^\times$. The characters of $(\ring{L}, +)$ have the properties listed below.
  \begin{enumerate}
    \item \emph{(Group homomorphism)} 
    
    \centerline{$\qquad \chi_y(a+b) = \chi_y(a)\chi_y(b)$;}
    \item \emph{(Symmetry)} 
    
    \centerline{$\qquad \chi_y(\lambda x) = \chi_1(\lambda xy) = \chi_x(\lambda y)$;}
    \item \emph{(Orthogonality of characters)} 
    
    \centerline{$\qquad \dfrac{1}{p^L}\displaystyle\sum_{y \in R/\pi^LR} \chi_y(x) = \delta_{0,x}$;}
    \item \emph{(Stability under the Galois group)}
    
    We have \[\chi_y(x) = \chi_{\sigma_u(y)}\big(\pi^{L-1}\sigma_u(x/\pi^{L-1})\big),\] where $\pi^{L-1}\sigma_u(x/\pi^{L-1}) \equiv u^{1-L}x \pmod{\pi R}$;
    \item \emph{(Characters of $\ring{L-1}$)}
    
    For $L\geq2$, the character $\chi_{\pi y} : \ring{L} \to \CC^\times$ induces a character on $\ring{L-1}$.
  \end{enumerate}
\end{proposition}

\begin{proof}
  We prove the above properties one by one.
\begin{itemize}
    \item Properties 1 and 2 follow directly from the definition of $\Phi_L$ in \Cref{pr:duality}.
    \item The orthogonality of characters holds trivially when $x = 0$. Suppose $x$ is not zero. Then there exists $z \in \ring{L}$ such that $\chi_z(x) \neq 1$. \mbox{Reindexing the character sum gives}
  \begin{align*}
    \frac{1}{p^L}\sum_{y \in \ring{L}} \chi_y(x) &= \frac{1}{p^L}\sum_{y \in \ring{L}} \chi_{y+z}(x) \\ &= \chi_z(x)\left(\frac{1}{p^L}\sum_{y \in \ring{L}} \chi_y(x)\right),
  \end{align*}
  which shows that the sum is zero.
    \item The stability under Galois automorphisms follows from the corresponding stability of the field trace. Let $x, y \in \ring{L}$. Since $\Tr \circ \sigma = \Tr$ for every $\sigma \in \operatorname{Gal}(K/\QQ)$, the exponent of the character satisfies
  \begin{align*}
    \Tr\left(\frac{xy}{p\pi^{L-1}}\right) &= \Tr\left(\sigma_u\left(\frac{xy}{p\pi^{L-1}}\right)\right) \\ &= \Tr\left(\frac{1}{p\pi^{L-1}}\Big[\pi^{L-1}\sigma_u\Big(\frac{x}{\pi^{L-1}}\Big)\Big]\sigma_u(y)\right).
  \end{align*}
  Moreover, we note that $\sigma_u(\pi)/\pi\equiv u\pmod{\pi R}$, hence $(\pi/\sigma_u(\pi))^{L-1}\equiv u^{1-L}\pmod{\pi R}$ and $\sigma_u(x) \equiv x \pmod{\pi R}$, which concludes the proof of this property.
    \item  Finally, the definition of $\Phi_L$ in \Cref{pr:duality} gives
  \begin{equation*}
    \chi_{\pi y}(x) = \exp\left(2i\pi \Tr\left(\frac{x\pi y}{p\pi^{L-1}}\right)\right) = \exp\left(2i\pi \Tr\left(\frac{xy}{p\pi^{L-2}}\right)\right),
  \end{equation*}
  which corresponds to the image of a character of $\ring{L-1}$.\qed
\end{itemize}
\end{proof}

In particular, the concept of characters allows us to define the Fourier transform on $\ring{L}$ and thus to correctly \mbox{define its quantum Fourier transform.}

\subsection{Construction of the Quantum Fourier Transform}

For odd $p$, the Fourier transform over Galois rings in~\cite{zhang_quantum_2009} does not directly cover the present ramified ring. Indeed,
\[\Phi_p(X)=(X-1)^{p-1}\pmod p.\]
We therefore use the trace-dual pairing above. Its Fourier transform can be implemented using ordinary abelian Fourier transforms \mbox{and a change of coordinates.}

\begin{definition}[Quantum Fourier Transform over $\ring{L}$]
  Let $p$ be a prime and $L\geq1$. We define the quantum Fourier transform (QFT) over $\ring{L}$ \mbox{as the linear operator specified by}
  \begin{equation*}
    \operatorname{QFT}_{\ring{L}}(\ket{x}) := \frac{1}{\sqrt{p^L}} \sum_{y \in R_L} \chi_y(x)\, \ket{y}.
  \end{equation*}
  for each $x\in \ring{L}$, and extend to all quantum states by linearity. Orthogonality of characters proves \mbox{that this map is unitary.}
\end{definition}

  We extend this quantum Fourier transform to $\ring{L}^n$ by performing the unitary $\operatorname{QFT}_{\ring{L}}$ coordinatewise. In other words, we can consider characters on $\ring{L}^n$ such that for $\x, \y \in \ring{L}^n$ we have,
  \[\chi_\y (\x) = \prod_{i = 1}^{n} \chi_{y_i} (x_i) = \exp\left(2 i \pi \Tr\left(\frac{\inner{\x}{\y}}{p \pi^{L-1}}\right)\right).\]

Note that when $L = 1$, we have $\ring{1} \cong \FF_p$, and $B_1(x,y)=-xy/p\bmod\ZZ$. Thus this QFT is the inverse of the positive-sign quantum Fourier transform \mbox{over $\FF_p$ introduced in Section~2.5.}

\begin{proposition}\label{le:construction_QFT}
  For fixed $p$ and $L=t(p-1)$, the quantum Fourier transform over $\ring{L}$ admits an implementation to error $\varepsilon$ in time $\operatorname{poly}_p(t,\log(1/\varepsilon))$.
\end{proposition}

\begin{proof}
  We first recall that $\ring{L} = \cycq$ when $L = (p-1)t$, and that $\cycq \cong \ZZ_q^{p-1}$. Hence, if we let $\iota : \ZZ_q^{p-1} \to \cycq$ denote such an isomorphism, the map
  \begin{equation*}
    \overline{B}_L : \ZZ_q^{p-1} \times \ZZ_q^{p-1} \to \QQ/\ZZ, \qquad \overline{B}_L({\mathbf{x}},{\mathbf{y}}) := B_L(\iota({\mathbf{x}}), \iota({\mathbf{y}})),
  \end{equation*}
  is a non-degenerate symmetric bilinear form. Consequently, there exists a matrix ${\mathbf{D}} \in \ZZ_q^{(p-1)\times(p-1)}$ such that
  \begin{equation*}
    \overline{B}_L({\mathbf{x}},{\mathbf{y}}) = \frac{1}{q}({\mathbf{x}}^T {\mathbf{D}} {\mathbf{y}}), \qquad \text{for all } {\mathbf{x}}, {\mathbf{y}} \in \ZZ_q^{p-1}.
  \end{equation*}
  In the power basis we take $D_{ij}=qB_L(\zetap^i,\zetap^j)\bmod q$; these entries are integers modulo $q$ since $q\ring{L}=0$. They are computable by rational arithmetic in $K$ in polynomial time in $t$ for fixed $p$. Non-degeneracy makes ${\mathbf{D}}$ invertible, and symmetry gives ${\mathbf{D}}^T={\mathbf{D}}$. Therefore, the unitary $\mathcal{U}_{\mathbf{D}} : \ket{{\mathbf{x}}} \mapsto \ket{\mathbf D\mathbf x}$ admits an efficient implementation, as does the quantum Fourier transform over $\ZZ_q^{p-1}$. The resulting composition acts as follows:
  \begin{align*}
    \operatorname{QFT}_{\ZZ_q^{p-1}} \circ\, \mathcal{U}_{\mathbf{D}} (\ket{{\mathbf{x}}}) &= \operatorname{QFT}_{\ZZ_q^{p-1}}(\ket{\mathbf D\mathbf x})\\ &= \frac{1}{\sqrt{q^{p-1}}} \sum_{{\mathbf{y}} \in \ZZ_q^{p-1}} \exp\!\left(\frac{2i\pi}{q}{\mathbf{x}}^T {\mathbf{D}} {\mathbf{y}}\right)\ket {\mathbf{y}} \\ &= \operatorname{QFT}_{\ring{L}}(\ket{{\mathbf{x}}}),
  \end{align*}
  which concludes the proof.\qed
\end{proof}

\subsection{Phase-state Problem}

In the spirit of Kuperberg's algorithm, the quantum Fourier transform on the ring $\ring{L}$ allows us to transform a CCP sample into a quantum state in which the information about the secret is encoded in the phase of the state. We therefore define a new problem, the phase-state problem (PSP), which consists in recovering the secret ${\mathbf{s}}$ from such phase-state samples.

\begin{definition}[Phase-States Problem]
  Let $n$ and $L$ be positive integers, and let $p$ be a prime. Fix a secret ${\mathbf{s}} \in \ring{L}^n$. Suppose we are given $\ell$ phase-state samples of the form $({\mathbf{y}}_i, \ket{\psi_{{\mathbf{y}}_i}})$, where
  \begin{equation*}
    \y_i \leftarrow \mathcal{U}\big(\ring{L}^n\big), \qquad \ket{\psi_{\y_i}} = \frac{1}{\sqrt p} \sum_{j \in \FF_p} \chi_{\y_i}(\lambda_j {\mathbf{s}})\, \ket{j}.
  \end{equation*}
  The labels $\y_i$ are independent. The \emph{phase-state problem}, denoted $\PSP_{n,L,p}^\ell$, asks to recover the secret ${\mathbf{s}}$.
\end{definition}

\begin{lemma}\label{le:CCP_to_PSP}
  Let $n$ be a positive integer, let $p$ be a prime, and let $q = p^t$ be a power of $p$. Fix $L := t(p-1)$. Then $\CCP_{n,q,p}$ reduces to $\PSP_{n,L,p}$.
\end{lemma}

\begin{proof}
  It is enough to show that one can build a phase-state sample from a CCP sample involving the same secret ${\mathbf{s}}$. A phase-state sample $({\mathbf{y}}_i, \ket{\psi_{{\mathbf{y}}_i}})$ with ${\mathbf{y}}_i \leftarrow \mathcal{U}\big(\ring{L}^n\big)$ is obtained from a CCP state $\ket{\phi_{\mathbf{x}}}$ by applying a quantum Fourier transform over $\cycq^n \cong \ring{L}^n$ to the second register. Indeed,
  \begin{align*}
    \big(\operatorname{id}_{\FF_p} \otimes \operatorname{QFT}_{\cycq}^{\otimes n}\big)(\ket{\phi_{\mathbf{x}}})
      &= p^{-(Ln+1)/2}\sum_{j \in \FF_p} \ket{j} \sum_{\y \in \ring{L}^n} \chi_\y(\x + \lambda_j \s)\, \ket{\y} \\
      &= p^{-(Ln+1)/2}\sum_{\y \in \ring{L}^n} \chi_\y(\x) \left(\sum_{j \in \FF_p} \chi_\y(\lambda_j \s)\, \ket{j}\right) \ket{\y} \\
      &= p^{-Ln/2}\sum_{\y \in \ring{L}^n} \chi_\y(\x)\, \ket{\psi_\y}\, \ket{\y}.
  \end{align*}
  Measuring the second register yields each $\y$ with probability $p^{-Ln}$, independently of $\x$. Its remaining factor $\chi_\y(\x)$ is a global phase. Independent input samples \mbox{therefore give independent uniform labels.}
  \qed
\end{proof}

\section{A practical algorithm for solving the CCP}
In 2023, Imran and Ivanyos \cite{Imran:2023tft} showed that there exists a quantum algorithm solving the hidden subgroup problem for nilpotent groups. They further point out that solving the hidden-shift problem via the "Simon-meets-Kuperberg" algorithm \cite{bonnetain_hidden_2018} is a special case of their algorithm \mbox{applied to the dihedral group.}

\begin{theorem}[\cite{Imran:2023tft}, Theorem 1.]\label{th:imran_ivanyos1}
	Let $G$ be a nilpotent group of bounded nilpotency class such that the prime factors of $|G|$ are constant. Assuming there exists a black box allowing a unique encoding of the elements via $\ell$-bit strings, there exists an exact quantum algorithm that solves the hidden-shift problem for $G$ using $\operatorname{poly}(\ell)$ operations and $\operatorname{poly}(\log |G|)$ calls to \mbox{the hidden subgroup problem oracle.}
\end{theorem}

Since, by Lemma~\ref{le:cyclotomic_grp_is_nilpotent}, the cyclotomic group $G_{n, q, p}$ is nilpotent, its nilpotency class is bounded by $\log(q) (p-1)$, and its cardinality is a power of $p$, Theorem~\ref{th:imran_ivanyos1} highlights an algorithm for solving the HSP for the cyclotomic group, i.e., a \mbox{potential algorithm for solving CCP.}

Nevertheless, some points need to be qualified. For applications in public key cryptography, we are interested in asymptotic results with respect to the dimension $n$ and the modulus $q$ of the problem. In this setting, the nilpotency class of the cyclotomic group is no longer constant. Moreover, the complexity studied is not a complexity tied to the number of oracle calls, but rather a time and space complexity. Indeed, the exact algorithm of Imran and Ivanyos assumes access to a unitary $U$ implementing state preparation and to its inverse $U^{-1}$. In our model, the input consists only of independent quantum samples: neither the preparation circuit nor its inverse is available. For these reasons, we have decided to present the algorithm of Imran and Ivanyos \cite{Imran:2023tft} \mbox{within the framework of CCP.}

It is worth mentioning that this connection was not immediately apparent to us. During our work on CCP, we independently arrived at an approach closely related to the algorithm of Imran and Ivanyos, before becoming aware of their work. The algorithm we had been developing can be viewed as an instance of their more general framework. We therefore present it here not as a new algorithm, but as a reformulation and specialization of their result to the cyclotomic setting, with particular emphasis on its concrete time and space complexity.\\

In this section, we present an algorithm for solving $\CCP$ in quasi-polynomial time using a quasi-polynomial number of samples. Combined with Proposition \ref{pro:DCP_and_CCP}, this yields a quasi-polynomial-time algorithm for DCP with arbitrary prime-power moduli. Our result generalizes the result of \cite{bai_quasi-polynomial_2025}.

\begin{figure}[H]
	\centering
	\begin{tikzpicture}[
		node distance=2 and 3.5, on grid,
		algo/.style={draw,},
		refe/.style={draw, align=center, densely dotted},
		label/.style={midway}
		]
		
		\node[algo] (CCP) {$\CCP_{n, q, p}$};
		\node[algo, right =of CCP] (PSP1) {$\PSP_{n, L, p}$};
		\node[algo, right =of PSP1] (PSP2) {$\PSP_{n, 1, p}$};
		\node[refe, right =of PSP2] (solver) {Solver\\ Section 4.3};
		
		\draw[->] (CCP) -- node[label, above] {Prop \ref{pro:DCP_and_CCP}} (PSP1);
		\draw[->] (PSP1) -- node[label, above] {Lemma \ref{le:sieving_step}} (PSP2);
		\draw[<-, densely dotted] (PSP2) -- node[label, above] {} (solver);
	\end{tikzpicture} 
	\caption{Overview of the CCP solver. The core idea is to reduce CCP to a phase-state problem and then perform sieving steps iteratively in order to obtain phase states from which information about the secret can be extracted.}
\end{figure}
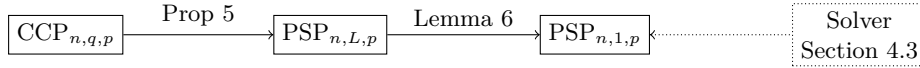

\subsection{Solver for the easy instance of $\PSP_{n,1,p}$}

Suppose we have a phase-state sample $(\y, \ket{\psi_\y})$ for an instance of $\PSP_{n,1,p}$ such that $\y \neq 0$. Since $\ring{1} = \FF_p$, we have $j \in R_1$, and the quantum state has the form
\begin{align*}
	\frac{1}{\sqrt{p}} \sum_{j \in \FF_p} \chi_\y(\lambda_j \s) \ket{j}
	&= \frac{1}{\sqrt{p}} \sum_{j \in \FF_p} \chi_\s(\lambda_j \y) \ket{j} \\
	&= \frac{1}{\sqrt{p}} \sum_{j \in \ring{1}} \chi_j\big(\inner{\y}{\s}\big) \ket{j} \\
	&= \operatorname{QFT}_{\ring{1}}\big(\ket{\inner{\y}{\s} \bmod p}\big).
\end{align*}
By measuring the phase state in its Fourier basis, we can recover $\inner{{\mathbf{y}}}{{\mathbf{s}}}$. Hence, if we have access to $n$ phase-state samples $({\mathbf{y}}_i, \ket{\psi_{{\mathbf{y}}_i}})$, with constant probability the ${\mathbf{y}}_i$ are linearly independent, we can then recover ${\mathbf{s}} \bmod p$. 

Hence, the goal is to reduce to this simple case using Lemma \ref{le:sieving_step}.

\subsection{Sieving step}

\begin{lemma}[from $\PSP_{n,  L, p}$ to $\PSP_{n, L-1, p}$]\label{le:sieving_step}
	Let $n$ and $L$ be positive integers, with $L\geq 2$. There exists a quantum algorithm which, given $r$ samples from $\PSP_{n, L, p}$ denoted by ${(\y_i, \ket{\psi_{\y_i}})}_{i = 1}^{r}$ with $\y_i$ in $R_L^n$ satisfying \(\sum_{i = 1}^r \y_i  = 0 \mod \pi R,\) is able to produce a sample from $\PSP_{n, L-1, p}$.
\end{lemma}

\begin{proof}
	Suppose that we have access to $r$ samples from $\PSP_{n, L, p}$, which we denote by ${(\y_i, \ket{\psi_{\y_i}})}_{i = 1}^{r}$. The following operations reduce the level by one.
	
	\begin{enumerate}
		\item We begin with the tensor product of the input states:
		\[ \bigotimes_{i = 1}^r \ket{\psi_{\y_i}} = \frac{1}{\sqrt{p^r}} \sum_{(j_1, \dots, j_r) \in \FF_p^{r}} \chi_\s \left(\sum_{i = 1}^{r} \y_i \lambda_{j_i}\right) \ket{j_1}\dots \ket{j_r}.\]
		\item For every $i\geq 2$, we perform the label change $j_i \mapsto j_i - j_1$ in order to introduce a dependence on $j_1$ in each register. This gives explicitly the following state:
		\[\frac{1}{\sqrt{p^r}} \sum_{(j_1, \dots, j_r) \in \FF_p^{r}} \chi_\s \left(\sum_{i = 1}^{r} \y_i \lambda_{j_i}\right) \ket{j_1} \ket{j_2 - j_1}\dots \ket{j_r-j_1}\]
		\item We then measure the last $r-1$ registers. Suppose that the output of the $i$-th register is $x_i \in \FF_p$. Fix $x_1 = 0$ and denote the free variable $j_1$ by $j$. In this case, the remaining state is:
		\[\frac{1}{\sqrt{p}} \sum_{j \in \FF_p} \chi_\s \left(\sum_{i  =1}^{r} \y_i \lambda_{x_i + j}\right)\ket{j}.\]
		Using the property $\lambda_{x_i + j} = \lambda_{x_i} + \zetap^{x_i}\lambda_j$, it follows that the remaining state is proportional to the state
		\[\frac{1}{\sqrt{p}} \sum_{j \in \FF_p} \chi_\s \left(\sum_{i  =1}^{r} \y_i \zetap^{x_i} \lambda_{j}\right)\ket{j}.\]
		We claim that the resulting state indeed corresponds to a $\PSP_{n, L-1, p}$.
	\end{enumerate}
	To this end, it suffices to observe that:
	\[\sum_{i = 1}^r \y_i \zetap^{x_i} = \sum_{i = 1}^{r} \y_i \mod \pi R^n = 0 \mod \pi R^n.\]
	Therefore, there exists $\y' \in \ring{L-1}^n$ such that $\sum_{i = 1}^r \y_i \zetap^{x_i} = \pi \y'$. The resulting state is thus of the form $(\y', \ket{\psi_{{\mathbf{y}}'}})$, a sample from $\PSP_{n, L-1, p}$.\qed
\end{proof}

In order to properly use the above lemma, it is necessary to efficiently find such a sequence $\y_i$ of vectors whose sum is zero modulo $\pi R^n$. The following theorem provides the required zero-sum search algorithm:

\begin{theorem}[\cite{Imran:2023tft}, Theorem 2]\label{th:imran_ivanos}
	There exists a deterministic algorithm which, given a sequence of size $S(n, p) = p^{O(p \log^2 p)} n^{O(p \log p)}$ of vectors in $\FF_p^n$, is able to find a subsequence whose sum is zero in time $\operatorname{poly}(S(n, p))$.
\end{theorem}

Observing that $\ring{1}^n = \FF_p^n$ , it suffices to apply the theorem to a sequence of size $S(n, p)$ of $\overline{\y}_i$ corresponding to the reduction modulo $\pi$ of the labels of the $\PSP_{n, L, p}$ samples of the form $(\y_i, \ket{\psi_{\y_i}})$.

\begin{lemma}[Uniform output labels]\label{le:fusion-uniform}
	Let $S=S(n,p)\geq2$ be an integer threshold for Theorem~\ref{th:imran_ivanos}. Given $S$ independent uniform labels in $\ring{L}^n$, select the nonempty zero-sum subsequence using only their residues modulo $\pi$, and apply Lemma~\ref{le:sieving_step}. The output label is uniform in $\ring{L-1}^n$. 
\end{lemma}
\begin{proof}
	Condition on the residues of all labels, choose fixed lifts ${\mathbf{r}}_i$, and write $\y_i={\mathbf{r}}_i+\pi {\mathbf{z}}_i$. The vectors ${\mathbf{z}}_i$ are independent and uniform in $\ring{L-1}^n$. The selected index set $I$ is now fixed. The measured differences $x_i$ are independent of the full labels by the preceding proof, so conditioning on them does not change the distribution of the ${\mathbf{z}}_i$. The output label is
	\[
	\y'={\mathbf{c}}+\sum_{i\in I}\zetap^{x_i}{\mathbf{z}}_i,
	\qquad {\mathbf{c}}=\frac{\sum_{i\in I}\zetap^{x_i}{\mathbf{r}}_i}{\pi}\in\ring{L-1}^n.
	\]
	The division is well defined because the numerator is divisible by $\pi$. Since $I$ is nonempty and every $\zetap^{x_i}$ is a unit, conditioning also on all but one of the ${\mathbf{z}}_i$ leaves $\y'$ uniform. \qed
\end{proof}

\subsection{Description of the sieving algorithm}

The sieving algorithm introduces lists $\calL_0, \dots, \calL_{L-1}$ such that, for every $i$, $\calL_i$ contains samples associated with $\PSP_{n, L-i, p}$. The main idea is to move from one level to the next by grouping in buckets of size $S(n, p)$, then applying the algorithm from \Cref{th:imran_ivanos} to find a zero-sum sequence, and doing the procedure described in Lemma~\ref{le:sieving_step}. Once the final stage is reached, we obtain a list $\calL_{L-1}$ consisting of samples from $\PSP_{n, 1, p}$, from which we can easily recover ${\mathbf{s}}_0 = {\mathbf{s}} \pmod{p}$. Once the secret modulo $p$ has been recovered, we can reduce $\CCP_{n, q, p}$ to an instance of $\CCP_{n, q/p, p}$ with secret ${({\mathbf{s}} - {\mathbf{s}}_0)}/{p}$ and repeat the same procedure until ${\mathbf{s}}$ is fully recovered.

The \mbox{algorithm is presented as follows.}

\begin{algorithm}
	\caption{CCP solver}\label{al:CCP_solver}
	\begin{algorithmic}[1]
		\Require $\ell$ CCP samples of a secret ${\mathbf{s}} \in \ZZ_q^n$
		\Ensure The secret ${\mathbf{s}} \in \ZZ_q^n$ of the CCP instance
		\State Build $\ell$ phase-state samples from the $\ell$ CCP samples, involving the same secret ${\mathbf{s}} \in \ZZ_q^n$.
		\State Regroup the phase-state samples into families of $S(n, p)$ elements and apply the algorithm from \Cref{th:imran_ivanos} to \mbox{each to get zero-sum sequences.}
		
		\State Apply \Cref{le:sieving_step} to obtain, on average, $\ell' \geq \ell \cdot \frac{1}{S(n, p)}$ phase-state samples associated with $\PSP_{n,L-1,p}$ with secret ${\mathbf{s}} \bmod \pi^{L-1} R$.
		\State Repeat Step 2 until obtaining, on average, $\ell \cdot (\frac{1}{S(n, p)})^{L-1}$ samples of $\PSP_{n,1,p}$.
		\State Deduce ${\mathbf{s}} \bmod p$ from this easy instance of $\PSP_{n,1,p}$.
		\State Repeat the previous steps for an instance of $\CCP_{1, q/p, p}$ with secret ${\mathbf{s}}_0 = {\mathbf{s}} \bmod (q/p)$, until the entire secret ${\mathbf{s}} \in \ZZ_q^n$ is recovered.
	\end{algorithmic}
\end{algorithm}

The analysis of Algorithm \ref{al:CCP_solver} shows that it runs in quasi-polynomial time when $q=\poly(n)$. This can be \mbox{summarized by the following theorem.}

\begin{theorem}[Quasi-Polynomial Algorithm for CCP]
	Let $n$ and $\ell$ be non-negative integers, let $p$ be a prime, and let $q$ be a power of $p$. When $\ell = 2^{\Omega(\log n \log q)}$, \Cref{al:CCP_solver} solves $\CCP_{n,q,p}^{\ell}$ with constant success probability in time $2^{{O}(\log n \log q)}$ using $\operatorname{poly}(n)$ quantum space.
\end{theorem}

\begin{proof}
	The first round (Steps 1--2--3--4) of the algorithm dominates the subsequent ones, and a total of $\log(q)$ rounds are performed. The overall complexity is therefore determined by the \mbox{complexity of the first round.}
	
	The goal is to find a lower bound on the size of the first list $\calL_0$ that ensures we have enough elements at level $\calL_{L-1}$ to allow us to recover $\s \mod p$. To do this, we note that at each reduction step, we obtain
	\[|\calL_{i+1}| \leq S(n, p)^{-1} |\calL_i|\]
	By applying this inequality $L-1$ times, we obtain
	\[|\calL_{L-1}| \leq S(n, p)^{-L + 1} |\calL_0|.\]
	
	On the other hand, we must ensure that there are sufficiently many elements at the final level to recover $\s \pmod{p}$. By Section 7.1, we know that $O(n)$ phase-states from $\PSP_{n, 1, p}$ are sufficient to recover $\s \mod p$.
	\[|\calL_{L-1}| \geq \operatorname{poly}(n)\]
	Combining these conditions on $|\calL_{L-1}|$, we obtain
	\[\poly(n) \geq S(n, p)^{-L + 1} |\calL_0| \Rightarrow |\calL_0| \geq (S(n, p))^{L} \operatorname{poly}(n).\]
	Thus, this justifies that the number of samples required is $\ell = 2^{\Omega(\log n \log q)}$, and it follows directly that the algorithm runs in $2^{O(\log n \log q)}$ time. The quantum space follows Regev’s construction \cite{regev_quantum_2004}.
	\qed
\end{proof}

\section{Equivalences between classical and cyclotomic problems}

\subsection{Between LWE and CyLWE}

The following proposition establishes a correspondence between search $\LWE_{n, q, \alpha}$ and cyclotomic $\CyLWE_{n, q, p, \alpha}$. Indeed, by showing that one can efficiently construct a sample of one problem from samples of the other, and vice versa, we can directly justify the \mbox{equivalence between the two problems.}

\begin{proposition}\label{th_LWE_to_CyLWE}
Let $p$ be a prime number, $q$ be a power of $p$, $\alpha>0$ be a scalar, and $n$ be an integer. There exists an algorithm that, given $p-1$ samples of $\LWE_{n, q, \alpha}$, produces a sample of $\CyLWE_{n, q, p, \alpha}$ involving the same secret ${\mathbf{s}}\in\ZZ_q^n$. Conversely, given a sample of $\CyLWE_{n, q, p, \alpha}$, there exists an algorithm that outputs $p-1$ samples of $\LWE_{n, q, \alpha}$.
\end{proposition}

\begin{proof}
  We first establish the forward direction.
  Suppose we are given an instance of classical $\LWE_{n, q, \alpha}$ with secret
  $\s \in \ZZ_q^n$.
  Take $(p-1)$ LWE samples $\{(\a_i, b_i)\}_{i=0}^{p-2}$ and form the following coefficient expansions:
  \[
    \a' := \sum_{i=0}^{p-2} a_i \zetap^i, \qquad
    b' := \sum_{i=0}^{p-2} b_i \zetap^i \pmod{q} = \inner{\a'}{\s} + e' \pmod{q}
  \]
  The pair $(\a', b')$ is an instance of $\CyLWE$ with secret $\s_0 \in \ZZ_q^n$ and error $e' = \sum_{i=0}^{p-2} e_i \zetap^i$, where $e_i$ denotes
  the error term in the sample $(\a_i, b_i)$.
  This error follows a distribution over $\cyc^n$ proportional to
  \[\prod_{i=0}^{p-2} \rho_{\alpha q} (e_i) = \rho_{\alpha q}((e_0, \dots, e_{p-2})) = \rho_{\alpha q} (e'),\]
   since the Gaussian weight factorize over coordinates, $D_{\ZZ [\zetap], \alpha}$ coincides with the product of $p-1$ independent one-dimensional discrete Gaussian.

  Conversely, suppose we have a sample of $\CyLWE_{n, q, p, \alpha}$ noted $(\a, b)$. We can decompose $\a$ and $b$ into the basis $(1, \zetap, \dots, \zetap^{p-2})$ such that
  \[
    \a = \sum_{i=0}^{p-2} \a_i \zetap^i, \qquad
    b = \sum_{i=0}^{p-2} b_i \zetap^i \pmod{q} = \sum_{i = 0}^{p-2} (\inner{\a_i}{\s} + e_i) \zetap^i \pmod{q},
  \]
  where $b_i = \inner{\a_i}{\s} + e_i \pmod{q}$ with $e_i \leftarrow D_{\ZZ, \alpha q}$.
  This yields $(p-1)$ samples of $\LWE$ with secret $\s$.
  \qed
\end{proof}

The introduction of CyLWE in the paper and its connection to LWE is extremely important to us, because it is what allowed us to fully understand how to link a classical problem to a problem on cyclotomic rings. Among other things, this is what enabled us to establish that the CCP secret must indeed be contained within the restricted space $\ZZ_q^n$ and not within $R_q^n$. All other reductions, like $\Cref{pro:DCP_and_CCP}$ and $\Cref{le:UEDCP_to_PSP}$, are inspired by the philosophy of the proof described above, which considers the group isomorphism $\ZZ_q^{p-1}$ to $R_q$.

\subsection{Between DCP and CCP}

The proof of the equivalence between DCP and CCP follows the same approach as the one used for the equivalence between LWE and CyLWE. Namely, we show that samples associated with one problem can be efficiently constructed from samples of the other problem, \mbox{while preserving the same secret.} 

\begin{proposition}\label{pro:DCP_and_CCP}
  Given $(p-1)$ samples of $\DCP_{n,q,p}$, there exists a quantum algorithm that produces a $\CCP_{n,q,p}$ sample with probability $p/2^{p-1}$. Conversely, given a sample of $\CCP_{n,q,p}$, there exists a quantum algorithm that outputs a $\DCP_{n,q}$ sample with probability $2/p$. Both direction involving the same secret.
\end{proposition}

\begin{proof}
  Let $f : \ZZ_q^{n\times(p-1)} \to \cycq^n$ denote the bijection defined, for a matrix ${\mathbf{X}} = (\x_0 \mid \cdots \mid \x_{p-2}) \in \ZZ_q^{n\times(p-1)}$, by
  \begin{equation*}
    f({\mathbf{X}}) := \sum_{i=0}^{p-2} \x_i \zetap^i.
  \end{equation*}

  We begin the first reduction by collecting $(p-1)$ samples of $\DCP_{n,q,p}$ of the form
  \[ \ket{\phi_i} \propto \sum_{j \in \{0,1\}} \ket{j}\ket{\x_i + j\s \bmod q}. \]
  Taking their tensor product and rearranging the registers gives the following expression for their joint state:
  \begin{align*}
    \bigotimes_{i=0}^{p-2} \ket{\phi_i} &= \sum_{(j_0,\dots,j_{p-2})\in \FF_2^{p-1}} \ket{j_0}\ket{\x_0+j_0 \s}\cdots\ket{j_{p-2}}\ket{\x_{p-2}+j_{p-2}{\mathbf{s}}} \\
    &= \sum_{{\mathbf{J}}\in \FF_2^{p-1}} \ket{{\mathbf{J}}}\ket{{\mathbf{X}}+{\mathbf{M}}_{{\mathbf{J}},{\mathbf{s}}}},
  \end{align*}
  with ${\mathbf{X}} = (\x_0 \mid \cdots \mid \x_{p-2}) \in \ZZ_q^{n\times(p-1)}$ and ${\mathbf{M}}_{{\mathbf{J}},\s} = (j_0 {\mathbf{s}} \mid \cdots \mid j_{p-2}\s) \in \ZZ_q^{n\times(p-1)}$. We then apply $f$ to the second register, obtaining
  \begin{align*}
    \sum_{{\mathbf{J}}\in\FF_2^{p-1}} \ket{{\mathbf{J}}}\ket{f({\mathbf{X}}+{\mathbf{M}}_{{\mathbf{J}},\s})} &= \sum_{{\mathbf{J}}\in\FF_2^{p-1}} \ket{{\mathbf{J}}}\ket{\x+\sum_{i\in\operatorname{supp}({\mathbf{J}})} \zetap^i \s},
  \end{align*}
  where $\x = \sum_{i=0}^{p-2} \x_i \zetap^i$ and $\operatorname{supp}({\mathbf{J}}) = \{0 \leq i \leq p-2 \mid J_i \neq 0\}$. This state is close to a $\CCP_{n,q,p}$ state. Indeed, a $\CCP_{n,q,p}$ state corresponds to the uniform superposition indexed by vectors ${\mathbf{J}}$ in the subset
  \begin{equation*}
    S = \{{\mathbf{J}}_j \in \FF_2^{p-1} \mid j \in \FF_p,\ ({\mathbf{J}}_j)_i = \mathbf{1}_{\{0,\dots,j-1\}}(i)\}.
  \end{equation*}
 Hence, by computing in an auxiliary register ${\mathbf{J}} \mapsto \mathbf{1}_S ({\mathbf{J}})$ and measuring it, we obtain the output $1$ with probability $p/2^{p-1}$. The remaining state corresponds to a superposition over $S$ from which we can convert to a $\CCP_{n, q, p}$ state:
  \[\sum_{{\mathbf{J}}_j \in S} \ket{{\mathbf{J}}_j}\ket{\x + \sum_{i = 0}^{j-1} \zetap^i \s} \to \sum_{j \in \FF_p} \ket{j}\ket{\x + \sum_{i = 0}^{j-1} \zetap^i \s}. \]

  Conversely, suppose we are given a $\CCP_{n,q,p}$ sample and we want to compute a $\DCP_{n,q}$ sample. 
  We measure the two outcome projector onto the first register subspace spanned by $\ket{0}$ and $\ket{1}$ and its orthogonal complement. With probability $2/p$, we get:
  \[\frac{1}{\sqrt{2}}\bigl(\ket{0}\ket{\x}+\ket{1}\ket{\x+\s}\bigr). \]
  We write  $\x=\sum_{i=0}^{p-2} x_i \zetap^i$ and apply the coefficient bijection $f^{-1}$ to the second register. Since $\s\in \cycq^n$ has only a constant coefficient (it's the coordinatewise constant embedding of an element of $\ZZ_q^n$ in $\cycq^n$), the state factors exactly as 
  \[
 		\frac{\ket{0} \ket{\x_0}+\ket{1}\ket{\x_0+ \s} }{\sqrt{2}} \otimes \ket{\x_1,\cdots,\x_{p-2}}
  \]
  We can discard the last registers \mbox{and get a DCP state.} \qed
  
%\yixin{il faut surement unifier les notations pour voir quand est-ce qu'on met du gras sur $x$ et $s$}

%  First, up to entangling $\ket{0,\dots,0}$, we compute $(j,0,\dots,0) \mapsto J_j$ on the first register and $f^{-1}$ on the second one:
%  \begin{align*}
%    \sum_{j \in \FF_p} \ket{j,0,\dots,0} \ket{x + \lambda_j s} \to \sum_{j \in \FF_p} \ket{J_j}\ket{X+M_{J_j,s}} = \sum_{J \in S} \ket{J} \ket{X+M_{J,s}}.
%  \end{align*}
%  In particular, $S$ contains $(0,\dots,0)$ and $(1,\dots,1)$ in $\FF_2^{p-1}$. By computing in an auxiliary register $J \mapsto \textbf{1}_{\{(0, \dots, 0), (1, \dots, 1)\}} (J)$ and measuring it, we get $1$ with probability bigger than $2^{-(p-1)}$ and the remaining state is
%  \begin{equation*}
%    \ket{0,\dots,0}\ket{X} + \ket{1,\dots,1}\ket{X+M_{(1,\dots,1),s}} \to \bigotimes_{i=0}^{p-2} \frac{1}{\sqrt{2}}\big(\ket{0}\ket{x_i} + \ket{1}\ket{x_i+s}\big).
%  \end{equation*}
%  
%  Hence, we obtain $(p-1)$ entangled $\DCP$ states. It remains to measure $(p-2)$ of these states to obtain one $\DCP$ sample.
%  \qed
\end{proof}

\subsection{A potential direct reduction from CyLWE to CCP?}

We believe it is possible to construct a direct reduction from CyLWE to CCP, analogous to Regev’s reduction from LWE to DCP \cite{regev_quantum_2004}.

Informally, the idea is that given a CyLWE $({\mathbf{A}}, {\mathbf{b}} = {\mathbf{A}}\s + {\mathbf{e}}) \in R_q^{m \times n} \times R_q^{m}$, for a secret $\s \in \ZZ_q^n$, we consider a superposition on the $q$-ary cyclotomic lattice:
\[\Lambda_q({\mathbf{A}}) = \{\y \in  R^m \mid \y= {\mathbf{A}}\x \mod q, \x \in R_q^n\},\]
which is shifted by a vector $\lambda_j {\mathbf{e}}$, where $j$ is in $\FF_p$. Given assumptions about the minimum distance of the $q$-ary lattice, the shift is very small compared to the minimum distance, and we can remove the error using cube-separation or ball-intersection techniques, as in \cite{brakerski_learning_2018}.

However, for all of this to work, we obviously need results on the topology of random $q$-ary lattice, which we do not claim to have. Furthermore, such a reduction would improve only very slightly upon a reduction consisting of moving from CyLWE to LWE, then from LWE to DCP via \cite{regev_quantum_2004}, and \mbox{finally from DCP to CCP.}

\section{Impact and discussions}

In this section, we discuss the impact of \Cref{al:CCP_solver} on other, more extensively studied computational problems such as EDCP or $S\ket{\LWE}$. This allows us to conclude that a quasi-polynomial-time algorithm exists when $q$ is \mbox{a power of a fixed prime $p$.}

\subsection{On the Extrapolated Dihedral Coset Problem.}

\begin{definition}[Search Extrapolated Dihedral Coset Problem \cite{brakerski_learning_2018}]
  Let a parameter $n$ be a dimension, $q\geq 2$ be a modulus and $M, \ell$ be positive integers. The uniform search Extrapolated Dihedral Coset Problem ($\operatorname{U-EDCP}_{n, q, M}^\ell$) consists of $\ell$ input states of the form
  \[\frac{1}{\sqrt{M}} \sum_{j = 0}^{M-1} \ket{j}\ket{\x_i + j \s \mod q}, \qquad \text{for } i = 0 \dots \ell,\]
  where $\x_i \in \ZZ_q^n$ are sampled uniformly, and asks to recover the secret $\s \in \ZZ_q^n$.
\end{definition}

\begin{lemma}[\cite{doliskani_efficient_2020}, Lemma~10]\label{le:EDCP_self_reduction}
  Let $n, q, \ell, M, M'$ be integers greater than $1$ and $M\geq M'$. There exists a probabilistic polynomial-time quantum reduction from $\operatorname{U-EDCP}_{n, q, M}^{\ell}$ to $\operatorname{U-EDCP}_{n, q, M'}^\ell$ that succeeds \mbox{with a constant success probability.}
\end{lemma}

\begin{theorem}
  Let $n, q, \ell, M$, be integers greater than $1$ and $q$ be any prime-power. There exists a quantum algorithm that solves $\operatorname{U-EDCP}_{n, q, M}^\ell$ in time $2^{O(\log n \log q)}$ using $\operatorname{poly}(n)$ quantum space, when $\ell = 2^{\Omega (\log n \log q)}$.
\end{theorem}

\begin{proof}
  The proof follows \Cref{le:EDCP_self_reduction} in order to get $\operatorname{U-EDCP}_{n, q, 2}^{\ell'}$ with $\ell' = O(\ell)$, then apply \Cref{pro:DCP_and_CCP} to reduce to a $\CCP_{n, q, p}^{O(\ell')}$ and use the CCP solver in \Cref{al:CCP_solver}.\qed
\end{proof}

\subsection{An algorithm for Gaussian $S\ket{\LWE}$}

In this section, we present a quasi-polynomial-time algorithm that solves the Gaussian $S\ket{\LWE}$ problem, based on the algorithm for solving $\CCP$ in Section 4. We thus generalize the procedure in \cite{bai_quasi-polynomial_2025} for this problem to any modulus $q = p^t$, where $p$ is a constant prime.
To do this, we first provide a definition of the Gaussian $S\ket{\LWE}$ problem and of $\overline{\operatorname{U-EDCP}}$. This second problem will allow us to bridge the gap between $S\ket{\LWE}$ and $\PSP$, defined in Section 3, to which we will be able to apply the procedure in \Cref{al:CCP_solver}.

\begin{definition}[Gaussian $S\ket{\LWE}$]
  Given parameters $\ell,n,q$ and a scalar $r>0$, Gaussian $S\ket{\LWE}$ consists of $\ell$ samples $(\a_i, \ket{\varphi_i})$ such that:
  \[\a_i \leftarrow \calU(\ZZ_q^n), \qquad \ket{\varphi_i} \propto \sum_{e \in \ZZ} \rho_r (e) \ket{\inner{\a_i}{\s} + e},\]
  and asks to recover the secret $\s \in \ZZ_q^n$.

\end{definition}

\begin{definition}[$\overline{\operatorname{U-EDCP}}$]
  Let a parameter $n$ be a dimension, $q \geq 2$ be a modulus and $M, \ell$ be positive integers. The $\overline{\operatorname{U-EDCP}}$ consists of $\ell$ \mbox{samples of the following form:}
  \[\y_k \leftarrow \calU(\ZZ_q^n), \qquad \ket{\psi_{\y_k}} = \frac{1}{\sqrt{M}} \sum_{j = 0}^{M-1} \omega_q^{j \inner{\y_k}{\s}}\ket{j},\]
  with $\omega_q = \exp\left(\frac{2 \sqrt{-1}\pi}{q}\right)$ and asks to recover the secret $\s \in \ZZ_q^n$.
\end{definition}

Next, we recall in \Cref{le:SLWE_to_UEDCP} the reduction that allows us to obtain $\overline{\operatorname{U-EDCP}}_{n, q, M}^{O(\ell)}$ samples from Gaussian $S\ket{\LWE}$ samples.

\begin{lemma}[from $S\ket{\LWE}$ to $\overline{\operatorname{U-EDCP}}$, \cite{bai_quasi-polynomial_2025}, Lemma~15~and~16] \label{le:SLWE_to_UEDCP}
  Let $\kappa$ be the security parameter, $\ell = \omega(\kappa)$ and $n, q = \operatorname{poly}(\kappa)$ be positive integers. Let $r = \Omega(\sqrt{\kappa})$ and $q/r = \Omega (\sqrt{\kappa})$. There exists a quantum polynomial-time reduction from $S\ket{\LWE}_{n, q, r}^\ell$ to $\overline{\operatorname{U-EDCP}}_{n, q, M}^{O(\ell)}$, where $M = c \cdot r$ for some constant $c$, \mbox{and succeeds with overwhelming probability.}
\end{lemma}

The next lemma converts these inputs into cyclotomic phase states.

\begin{lemma}[from $\overline{\operatorname{U-EDCP}}$ to $\PSP$]\label{le:UEDCP_to_PSP}
  Let $n, q, M, p$ be integers greater than $1$, with $p$ prime and $q = p^t$. There exists a quantum algorithm that given a constant number of samples $\overline{\operatorname{U-EDCP}}_{n, q, q}$, outputs $\PSP_{n, L, p}$ samples, with $L = (p-1)\log q$, involving the same secret $\s \in \ZZ_q^n$ with constant probability.
\end{lemma}

\begin{proof}
We describe the successive steps that allow us to obtain phase-state samples from states of the following form. We omit normalization factors \mbox{in the intermediate expressions below.}
\[\ket{\phi_i} \propto \sum_{x_i \in \ZZ_q} \omega_q^{x_i \inner{\y_i}{\s}} \ket{x_i}, \qquad \text{ for } i = 0 \dots p-2,  \]
 with known labels $\y_i$.
\begin{enumerate}
\item We begin by entangling the $p-1$ states to obtain the following state:
\begin{align*}
\bigotimes_{i = 0}^{p-2} \ket{\phi_i} &= \sum_{({\mathbf{x}}_0, \dots, {\mathbf{x}}_{p-2})\in \{0, \dots, M-1\}^{p-1}} \omega_q^{\sum_{i =0}^{p-2} x_i \inner{\y_i}{\s}} \ket{{\mathbf{x}}_1}\dots \ket{{\mathbf{x}}_{p-2}}\\
&= \sum_{{\mathbf{x}} \in \{0, \dots, M-1\}^{p-1}} \omega_q^{{\mathbf{x}}^Tv} \ket{{\mathbf{x}}},
\end{align*}
where ${\mathbf{v}} \in \ZZ_q^{p-1}$ is such that its $i$-th coordinate is $v_i = \inner{\y_i}{\s}$.
\item Let ${\mathbf{D}} \in \ZZ_q^{(p-1)\times (p-1)}$ be the invertible matrix associated with the bilinear form $B_L$ from \Cref{le:construction_QFT} with $L = t(p-1)$. By noting ${\mathbf{v}}' = {\mathbf{D}}^{-1} {\mathbf{v}}$, in other words there exists $(\y'_i)_{i = 0}^{p-2}$ such that ${\mathbf{v}}'_i = \inner{\y'_i}{\s}$ for $0\leq i \leq p-2$. We have
\begin{align*}
  \sum_{{\mathbf{x}} \in \{0, \dots, M-1\}^{p-1}} \omega_q^{{\mathbf{x}}^T {\mathbf{D}} {\mathbf{v}}'} \ket{{\mathbf{x}}} = \sum_{{\mathbf{x}} \in \iota(\{0, \dots, M-1\}^{p-1})} \chi_{\mathbf{x}} (\iota({\mathbf{v}}')) \ket{{\mathbf{x}}}
\end{align*}
where $\iota: \ZZ_q^{p-1} \to R_q$ is the injection with respect to the basis $(1, \zeta_p, \dots, \zetap^{p-2})$. Moreover, $\iota({\mathbf{v}}') = \sum_{i = 0}^{p-2} \zetap^{i} \inner{\y'_i}{\s} = \inner{\y'}{{\mathbf{s}}}$, where $\y' = \sum_{i = 0}^{p-2} \zetap^i \y'_i$. Note that $\y'$ follows a uniform distribution on $\iota(\{0, \dots, M-1\}^{p-1})$ by invertibility of the coordinate transformation. We obtain
\[\chi_{\mathbf{x}} (\iota({\mathbf{v}})) = \exp\left(2 i \pi \Tr\left(\frac{\inner{{\mathbf{x}} \y'}{\s}}{p \pi^{L-1}}\right)\right) = \chi_\s ({\mathbf{x}} \y').\]
The resulting state at the end of this operation is therefore
\[\sum_{{\mathbf{x}} \in \iota(\{0, \dots, M-1\}^{p-1})} \chi_\s ({\mathbf{x}} \y') \ket{{\mathbf{x}}}.\]
\item The idea is now to move from a superposition over ${\mathbf{x}} \in \iota\left(\{0, \dots, M-1\}^{p-1}\right)$ to a superposition over $E := \{\lambda_j\}_{j \in \FF_p}$ and thereby obtain a $\PSP$ sample. To this end, we will use the intermediate superposition $\iota(\{0, 1\}^{p-1})$. We therefore compute in an auxiliary register the function ${\mathbf{x}} \mapsto \lfloor\iota^{-1} ({\mathbf{x}})/2 \rfloor$, which we measure directly. Let ${\mathbf{k}} \in \ZZ_q^{p-1}$ denote the output, and suppose that ${\mathbf{k}} \in \{0, \dots, (M-1)/2 \}^{p-1}$ this happens with probability bounded below by a positive constant when $M$ is large (otherwise, we discard the resulting state and restart from the beginning). The remaining state has support \mbox{in the following translated set:}
\[\iota({\mathbf{k}}) + \iota(\{0, 1\}^{p-1}) \subset \iota(\{0, \dots, M-1\}^{p-1}),\]
that is,
\[\sum_{{\mathbf{x}} \in \iota(\{0, 1\}^{p-1})} \chi_\s (({\mathbf{x}} - \iota({\mathbf{k}})) \y') \ket{{\mathbf{x}}} \propto \sum_{{\mathbf{x}} \in \iota(\{0, 1\}^{p-1})} \chi_\s ( {\mathbf{x}} \y') \ket{{\mathbf{x}}}. \]
\item Finally, it suffices to compute the function ${\mathbf{x}} \mapsto \mathbf{1}_{E} ({\mathbf{x}})$ in a new auxiliary register. After measurement, there is a constant probability $p/2^{p-1}$ of obtaining the desired output $1$ and thus obtaining a $\PSP$ state with label $\y'$ and with respect to the same secret $\s \in \ZZ_q^{p-1}$. \qed
\end{enumerate}
\end{proof}

Observe that this lemma provides an alternative way to prove the existence of a quasi-polynomial-time algorithm for uniform EDCP without having to resort \mbox{to a reduction to DCP.}

\begin{theorem}
  Let $\kappa$ be the security parameter and $n, q = \poly(\kappa)$ be integers, where $q = p^t$ is a power of a fixed prime $p$. Let $r = \Omega(\sqrt{\kappa})$ and $q/r = \Omega(\sqrt{\kappa})$. There exists a quantum algorithm that solves $S \ket{\LWE}^\ell_{n, q, r}$ in time $2^{O(\log^2 n)}$ using $\poly(n)$ quantum space, when $\ell = 2^{\Omega(\log^2 n)}$.
\end{theorem}

\begin{proof}
  Start by applying \Cref{le:SLWE_to_UEDCP} to obtain $O(\ell)$ $\overline{\operatorname{U-EDCP}}$ samples, then applying \Cref{le:UEDCP_to_PSP} to get $O(\ell)$ phase-state sample involving the same secret $\s \in \ZZ_q^n$. Then it remains to apply the procedure from \Cref{al:CCP_solver} to recover $\s$ from a \mbox{quasi-polynomial number of phase-state samples.}\qed
\end{proof}

Note that the \Cref{le:UEDCP_to_PSP}, which is used to prove the existence of an algorithm for the Gaussian $S\ket{\LWE}$, is specific to our paper. Indeed, to the best of our knowledge, $S\ket{\LWE}$ is not a hidden subgroup problem. We therefore would not have been able to obtain such an algorithm had we not defined a quantum Fourier transform on $R_q$ and had we not \mbox{explicitly specified the phase-state samples.}

\paragraph{Non-impact on Learning with Errors.}  Since this result generalizes the algorithm in \cite{bai_quasi-polynomial_2025}, we obtain the same conclusion regarding the security of LWE. Indeed, the sieving algorithm requires too many CCP samples, and the reduction from LWE to DCP places an upper bound on the number of samples that is smaller than the number of samples required to solve the problem.

\bibliographystyle{splncs04}
\bibliography{references.bib}

\end{document}